\documentclass{birkjour}

\usepackage[T1]{fontenc}

\usepackage{amsmath,amsthm,amssymb}

\def\R{\mathbb R}
\def\N{\mathbb N}
\def\C{\mathbb C}

\def\Q{\mathbb Q}

\newtheorem{lemma}{Lemma}[section]
\newtheorem{proposition}[lemma]{Proposition}
\newtheorem{theorem}[lemma]{Theorem}
\newtheorem{corollary}[lemma]{Corollary}
\newtheorem{observation}[lemma]{Observation}

\theoremstyle{definition}

\renewcommand{\d}{{\mathrm d}}
\renewcommand{\i}{{\mathrm i}}
\def\e{\mathrm e}

\def\O{\mathcal O}

\def\H{\mathcal{H}}
\def\V{\mathcal{V}}

\DeclareMathOperator{\tg}{tg}
\DeclareMathOperator{\cotg}{cotg}
\DeclareMathOperator{\sgn}{sgn}

\begin{document}

\title{Spectrum of a dilated honeycomb network}

%----------Author 1
\author[Exner]{Pavel Exner}

\address{%
Nuclear Physics Institute \\
Academy of Sciences of the Czech Republic \\
Hlavn\'{\i} 130,  250 68 \v Re\v{z}, Czech Republic, and \\
Doppler Institute \\
Czech Technical University \\
B\v{r}ehov\'{a} 7, 11519 Prague, Czech Republic}

\email{exner@ujf.cas.cz}

%\thanks{The work was supported by the Czech Science Foundation under the project 14-06818S}
%----------Author 2
\author[Turek]{Ond\v{r}ej Turek}
\address{Nuclear Physics Institute \\
Academy of Sciences of the Czech Republic \\
Hlavn\'{\i} 130,  250 68 \v Re\v{z}, Czech Republic, and \\
Bogolyubov Laboratory of Theoretical Physics \\
Joint Institute for Nuclear Research \\
141980 Dubna, Russia}

\email{turek@theor.jinr.ru}


\begin{abstract}

We analyze spectrum of Laplacian supported by a periodic honeycomb lattice with generally unequal edge lengths and a $\delta$ type coupling in the vertices. Such a quantum graph has nonempty point spectrum with compactly supported eigenfunctions provided all the edge lengths are commensurate. We derive conditions determining the continuous spectral component and show that existence of gaps may depend on number-theoretic properties of edge lengths ratios. The case when two of the three lengths coincide is discussed in detail.

\end{abstract}

\maketitle

\section{Introduction}

Quantum graphs, more exactly differential operators on metric graphs describing quantum motion confined to networks, attracted a lot of attention recently as a fruitful combination of spectral theory, geometry, combinatorics, and other disciplines. The number of results in this area is large and permanently increasing; we refer to the monograph \cite{BK} for an up-to-date survey.

A class of particular interest are quantum graphs having a periodic structure. On one hand they are interesting mathematically, in particular, because the corresponding operators may exhibit properties different from standard periodic Schr\"odinger operators, for instance they may have compactly supported eigenfunctions. On the other hand, they provide a physical model of various systems having crystalline structure which become popular especially recently in connection with the discovery of graphene and related material objects such as carbon nanotubes \cite{KP07}.

Physical models of various lattice structures usually involve symmetries as arrangements which the nature favours. This may be true in the ideal situation but it can change under influence of external forces, for instance, mechanical strains. At the same time, we know from the simple model of a rectangular quantum-graph lattice \cite{Ex96, EG96} that the graph geometry may give rise to interesting number-theoretic effects in the spectrum. This motivates us to inspect how edge length variations can affect the spectrum of the lattice appearing most frequently in the applications, the hexagonal one.

Let us thus consider an infinite honeycomb graph $\Gamma$ dilated independently in all the three directions, as depicted in Figure~\ref{Hexagon} below. That is, each hexagon consists of two andtipodal edges of length $a$, two antipodal edges of length $b$ and two antipodal edges of length $c$. The operator to investigate is the corresponding quantum-graph Hamiltonian, that is, a Laplacian on the Hilbert space $\H=L^2(\Gamma)$ consisting of sequences $\psi=\{\psi_j\}$ the elements of which refer to edges of $\Gamma$. The operator acts as $H\psi=\{-\psi_j''\}$ on functions from $H^1(\Gamma) \cap H^2(\Gamma\setminus\V)$, where $\V$ is the set of graph vertices. In order to make it self-adjoint we have to specify its domain, for instance, by indicating boundary conditions. We choose the so-called $\delta$-coupling \cite{Ex96} requiring
 % ------------- %
\begin{equation} \label{bc-delta}
\psi_1(0+)=\psi_2(0+)=\psi_3(0+)=:\psi(0)\,, \qquad \sum_{i=1}^3 \psi'_i(0+)= \alpha \psi(0)\,,
\end{equation}
 % ------------- %
where $i=1,2,3$ number three edges meeting in a vertex, which are para\-metri\-zed by their arc length with zero at the junction. We suppose that the coupling is the same at each vertex, hence the operators exhibit translational symmetry corresponding to the geometry of the hexagonal lattice. It would be thus natural to label the operator by the parameter appearing in (\ref{bc-delta}) writing it, for instance, as $H_\alpha$ for a fixed $\alpha\in\R$, however, since there will be no danger of misunderstanding, we shall drop the index.

An alternative way is to characterize the operator $H$ by means of the associated quadratic form which is given by
 % ------------- %
\begin{equation} \label{qform}
q[\psi] = \int_\Gamma |\psi'(s)|^2 \d s + \alpha \sum_i |\psi_i|^2
\end{equation}
 % ------------- %
with the domain consisting of all functions from $H^1(\Gamma)$, where the first term is a shorthand for the sum over all the edges and in the second term we sum over all the vertices and $\psi_i$ is the function value at the $i$-th vertex. It is obvious from (\ref{qform}) that $H\ge 0$ holds for $\alpha\ge 0$, and it is not difficult to check that for $\alpha<0$ we have $\inf\sigma(H)<0$.

Our goal in this paper is to analyze the spectrum of $H$. Since the system is periodic, it has a band structure but in general it can have a nonempty point component \cite{BK}. We are going to show that this happens \emph{iff} all the lattice edges are commensurate. Next we derive the condition determining the spectrum, in particular, its open gaps. After the general discussion, we focus in Sec.~\ref{s: b=c} on the particular case when two of the three edge lengths are identical and analyze the gap structure in detail. We conclude the paper by mentioning a couple of  questions about the model which remain open.

\section{Point spectrum}

In contrast to the usual Schr\"odinger operator theory, quantum graph Hamiltonians may violate the unique continuation principle -- see, e.g., \cite{Ku05}. It happens also in our present model; a sufficient condition for that is a commensurability of the lattice edges lengths.

 % -------------- %
\begin{proposition}\label{Bod postacujici}
If $\frac{b}{a}\in\Q$ and $\frac{c}{a}\in\Q$, then $\sigma_\mathrm{p}(H)\ne\emptyset$.
\end{proposition}
 % -------------- %
\begin{proof}
Under the assumption, there is an infinite number of values $k$ such that $ka$, $kb$, and $kc$ are integer multiples of $2\pi$. Then a sinusoidal function on a perimeter of a hexagon cell with zeros at the vertices gives rise to an eigenfunction of $H$ since it solves the equation $-\phi''=k^2\phi$ and satisfies the boundary conditions \eqref{bc-delta}.
\end{proof}

\noindent It is also obvious that such a point spectrum is of infinite multiplicity. On the other hand, the commensurability is also a necessary condition.

 % -------------- %
\begin{proposition}\label{Bod nutna}
If $\sigma_\mathrm{p}(H)\ne\emptyset$, then $\frac{b}{a}\in\Q$ and $\frac{c}{a}\in\Q$.
\end{proposition}
 % -------------- %

\noindent We postpone the proof of this claim to the next section.

\section{Continuous spectrum}

\subsection{Determining the spectrum}

Since we are dealing with a periodic graph, a natural tool to employ is  the Floquet-Bloch decomposition \cite[Chap.~4]{BK}. The elementary cell of $\Gamma$ is shown in Fig.~\ref{Hexagon}, together with the symbols we use to denote the wave function components on the edges.

 % -------------- %
\begin{figure}[h]
\begin{center}
\includegraphics[angle=0,width=0.7\textwidth]{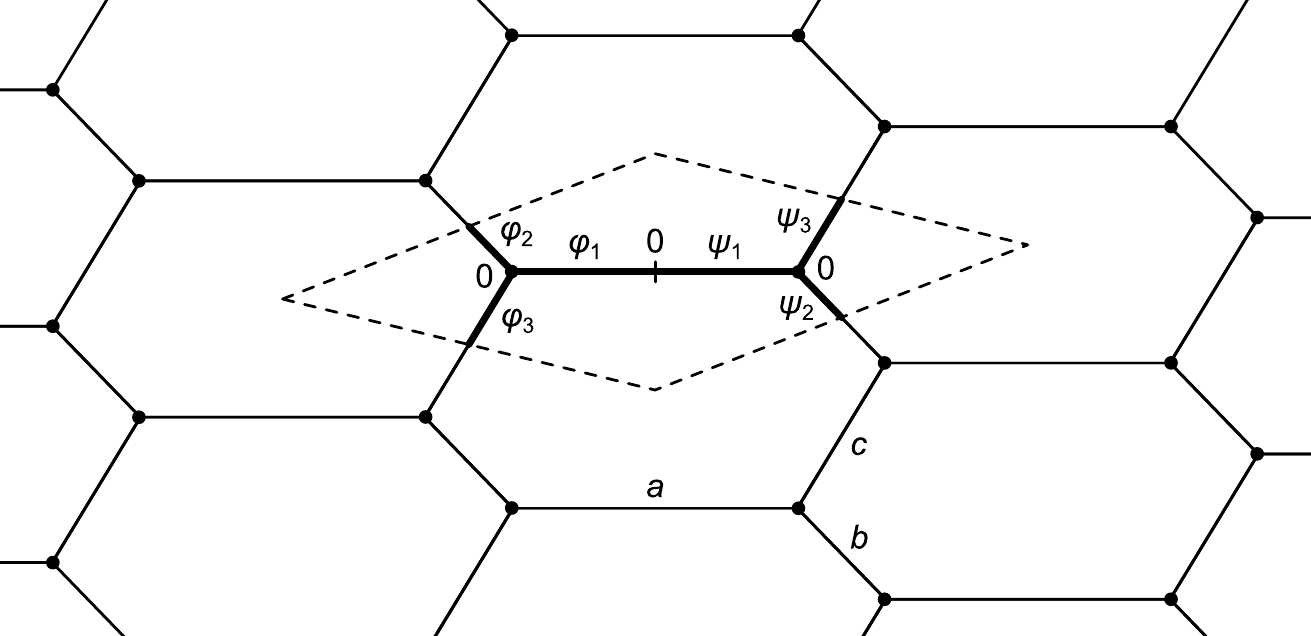}
\caption{A dilated honeycomb network and the elementary cell}\label{Hexagon}
\end{center}
\end{figure}
 % -------------- %

We are interested in generalized eigenfunctions of the graph Laplacian at an energy $E$. If $E>0$, we put conventionally $E=k^2$ with $k>0$ and assume that $\sin(\ell k)\neq 0$ holds for at least one $\ell\in\{a,b,c\}$; without loss of generality we may suppose that $\sin(ak)\neq0$. Since the Hamiltonian acts as a negative second derivative, the wavefunction on each edge has to be a linear combination of the exponentials $\e^{\i kx}$ and $\e^{-\i kx}$, specifically we can write
 % -------------- %
\begin{subequations}\label{vlnfce}
\begin{align}
\psi_1(x)&=C_1^+\e^{\i k x}+C_1^-\e^{-\i k x},\quad x\in[0,a/2] \label{psi1} \\
\psi_2(x)&=C_2^+\e^{\i k x}+C_2^-\e^{-\i k x},\quad x\in[0,b/2] \label{psi2} \\
\psi_3(x)&=C_3^+\e^{\i k x}+C_3^-\e^{-\i k x},\quad x\in[0,c/2] \label{psi3} \\
\varphi_1(x)&=D_1^+\e^{\i k x}+D_1^-\e^{-\i k x},\quad x\in[-a/2,0] \label{phi1} \\
\varphi_2(x)&=D_2^+\e^{\i k x}+D_2^-\e^{-\i k x},\quad x\in[-b/2,0] \label{phi2} \\
\varphi_3(x)&=D_3^+\e^{\i k x}+D_3^-\e^{-\i k x},\quad x\in[-c/2,0] \label{phi3}
\end{align}
\end{subequations}
 % -------------- %
Obviously, $\psi_1(0)=\varphi_1(0)$ and $\psi_1'(0)=\varphi_1'(0)$, hence
 % -------------- %
\begin{equation}\label{CD1}
C_1^+=D_1^+\,,\qquad C_1^-=D_1^-\,.
\end{equation}
 % -------------- %
The wave functions have to satisfy the following six boundary conditions corresponding to the $\delta$-couplings in the vertices (\ref{bc-delta}), namely
 % -------------- %
\begin{subequations}\label{delta}
\begin{gather}
\psi_2(0)=\psi_3(0)=\psi_1(a/2) \label{cont.psi} \\
\psi_2'(0)+\psi_3'(0)-\psi_1'(a/2)=\alpha\psi_1(0) \label{der.psi} \\
\varphi_2(0)=\varphi_3(0)=\varphi_1(-a/2) \label{cont.phi} \\
-\varphi_2'(0)-\varphi_3'(0)+\varphi_1'(-a/2)=\alpha\varphi_1(0) \label{der.phi}
\end{gather}
\end{subequations}
 % -------------- %
where $\alpha\in\R$ is the coupling parameter. On the other hand, the Floquet-Bloch decomposition requires to impose the following conditions,
 % -------------- %
\begin{equation}\label{Floquet}
\begin{split}
\psi_2(b/2)=\e^{\i\theta_1}\varphi_2(-b/2)\,,&\qquad\qquad \psi_3(c/2)=\e^{\i\theta_2}\varphi_3(-c/2)\,,\\
\psi_2'(b/2)=\e^{\i\theta_1}\varphi_2'(-b/2)\,,&\qquad\qquad \psi_3'(c/2)=\e^{\i\theta_2}\varphi_3'(-c/2)
\end{split}
\end{equation}
 % -------------- %
for some $\theta_1,\theta_2\in(-\pi,\pi]$.
%Substituting~\eqref{psi2}, \eqref{psi3}, \eqref{phi2}, \eqref{phi3}
Substituting~\eqref{psi2}--\eqref{phi3} into~\eqref{Floquet} enables one to express variables $D_2^{\pm}$ and $D_3^{\pm}$ in terms of $C_2^{\pm}$ and $C_3^{\pm}$: we obtain
 % -------------- %
\begin{equation}\label{elim}
\begin{split}
D_2^+=C_2^+\cdot\e^{\i(bk-\theta_1)}\,,&\qquad\qquad D_3^+=C_3^+\cdot\e^{\i(ck-\theta_2)}\,, \\
D_2^-=C_2^-\cdot\e^{\i(-bk-\theta_1)}\,,&\qquad\qquad D_3^-=C_3^-\cdot\e^{\i(-ck-\theta_2)}\,.
\end{split}
\end{equation}
 % -------------- %
The continuity at the vertices -- cf. conditions~\eqref{cont.psi}, \eqref{cont.phi} -- together with \eqref{CD1} allow us to eliminate coefficients $C_1^\pm$ and $D_1^\pm$. In this way we obtain a system of four linear equations containing $C_2^+,C_2^-,C_3^+,C_3^-$ as the unknown quantities and $a,b,c,k,\alpha$ as parameters, namely
 % -------------- %
\begin{equation}\label{matice}
M\left(\begin{array}{c}C_2^+\\C_2^-\\C_3^+\\C_3^-\end{array}\right)=0\,,
\end{equation}
 % -------------- %
where the matrix $M$ is given as
 % -------------- %
$$
M=\left(\begin{array}{cccc}
1 & 1 & -1 & -1\\
\e^{\i(bk-\theta_1)} & \e^{\i(-bk-\theta_1)} & -\e^{\i(ck-\theta_2)} & -\e^{\i(-ck-\theta_2)} \\
m_{31} & m_{32} & \i & -\i \\
m_{41} & m_{42} & -\i\e^{\i(ck-\theta_2)} & \i\e^{\i(-ck-\theta_2)}
\end{array}\right)
$$
 % -------------- %
with
 % -------------- %
$$
m_{3j}:= \frac{-\e^{-\i\sigma_j ak}+\e^{\i(\sigma_j bk-\theta_1)}}{\sin ak}-\frac{\alpha}{k}
$$
 % -------------- %
and
 % -------------- %
$$
m_{4j}:= \frac{-\e^{\i(\sigma_j ak+ \sigma_j bk-\theta_1)}+1}{\sin ak}-\frac{\alpha}{k}\e^{\i(\sigma_j bk-\theta_1)}
$$
 % -------------- %
for $j=1,2$, where $\sigma_j:= (-1)^{j-1}$. A nontrivial solution of the form~\eqref{vlnfce} exists \emph{iff}  $\left(C_2^+,C_2^-,C_3^+, C_3^-\right)$ is a nonzero vector. Therefore, $k^2$ belongs to the spectrum of $H$ if~\eqref{matice} has a non-trivial solution for certain pair $(\theta_1,\theta_2)$, in other words, if there exist $\theta_1,\theta_2\in(-\pi,\pi]$ such that $\det(M)=0$. A~straightforward calculation leads to
 % -------------- %
\begin{equation}\label{det(M)}
\begin{split}
& \det(M) =-4\Big[
2\sin ak\cos bk\cos ck+2\cos ak\sin bk\cos ck \\
&\quad +2\cos ak\cos bk\sin ck -3\sin ak\sin bk\sin ck \\
&\quad -2\sin ak\cos(\theta_1-\theta_2)-2\sin ck\cos\theta_1-2\sin bk\cos\theta_2 \\
&\quad +2\frac{\alpha}{k}(\cos ak\sin bk\sin ck+\sin ak\cos bk\sin ck+\sin ak\sin bk\cos ck) \\
&\quad +\frac{\alpha^2}{k^2}\sin ak\sin bk\sin ck\Big]\frac{\e^{-\i(\theta_1+\theta_2)}}{\sin ak}\,.
\end{split}
\end{equation}
 % -------------- %
The spectral condition can be put into a more convenient form if we exclude all the ``Dirichlet points'', i.e. if we consider $k$ such that $\sin(\ell k)\neq 0$ holds for all $\ell\in\{a,b,c\}$. After a simple manipulation, we then obtain
 % -------------- %
\begin{equation*}
\begin{split}
\det(M) &
=-4 \bigg[
2(\cotg ak\cotg bk+\cotg ak\cotg ck+\cotg bk\cotg ck) \\
& +\cotg^2 ak+\cotg^2 bk+\cotg^2 ck-\frac{1}{\sin^2 ak}-\frac{1}{\sin^2 bk}-\frac{1}{\sin^2 ck} \\
& -2\left(\frac{\cos\theta_1}{\sin ak\sin bk}+\frac{\cos\theta_2}{\sin ak\sin ck}+\frac{\cos(\theta_1-\theta_2)}{\sin bk\sin ck}\right) \\
& +2\frac{\alpha}{k}(\cotg ak+\cotg bk+\cotg ck)+\frac{\alpha^2}{k^2}\bigg]\frac{\sin bk\sin ck}{\e^{\i(\theta_1+\theta_2)}}\,,
\end{split}
\end{equation*}
 % -------------- %
hence
 % -------------- %
\begin{equation*}
\begin{split}
\det(M)= & -4\left[
\left(\cotg ak+\cotg bk+\cotg ck+\frac{\alpha}{k}\right)^2 \right. \\
& -\frac{1}{\sin^2 ak}-\frac{1}{\sin^2 bk}-\frac{1}{\sin^2 ck} \\ & \hspace{-3em} \left.- 2\left(\frac{\cos\theta_1}{\sin ak\sin bk}+\frac{\cos\theta_2}{\sin ak\sin ck}+\frac{\cos(\theta_1-\theta_2)}{\sin bk\sin ck}\right)\right]\frac{\sin bk\sin ck}{\e^{\i(\theta_1+\theta_2)}}\,.
\end{split}
\end{equation*}
 % -------------- %
We can conclude that $k^2\in\sigma(H)$ holds if there are $\theta_1,\theta_2\in(-\pi,\pi]$ such that
 % -------------- %
\begin{multline}\label{SP theta}
\left(\cotg ak+\cotg bk+\cotg ck+\frac{\alpha}{k}\right)^2
=\frac{1}{\sin^2 ak}+\frac{1}{\sin^2 bk}+\frac{1}{\sin^2 ck} \\ +2\left(\frac{\cos\theta_1}{\sin ak\sin bk}+\frac{\cos\theta_2}{\sin ak\sin ck}+\frac{\cos(\theta_1-\theta_2)}{\sin bk\sin ck}\right)\,.
\end{multline}
 % -------------- %
The obtained spectral condition allows us to determine the positive part of the spectrum. This is sufficient if $\alpha\ge 0$, in the opposite case we have to take also negative energies into account. This can be done in a similar way, replacing the positive $k$ in the above considerations by $k=\i\kappa$ with $\kappa>0$. In particular, the condition (\ref{SP theta}) is then replaced by
 % -------------- %
\begin{multline}\label{SP theta neg}
\left(\coth a\kappa+\coth b\kappa+\coth c\kappa+\frac{\alpha}{\kappa}\right)^2
=\frac{1}{\sinh^2 a\kappa}+\frac{1}{\sinh^2 b\kappa}+\frac{1}{\sinh^2 c\kappa} \\ +2\left(\frac{\cos\theta_1}{\sinh a\kappa\sinh b\kappa}+\frac{\cos\theta_2}{\sinh a\kappa\sinh c\kappa}+\frac{\cos(\theta_1-\theta_2)}{\sinh b\kappa\sin c\kappa}\right)\,;
\end{multline}
 % -------------- %
in distinction to the previous case there is no need to exclude any values of the spectral parameter $\kappa$.

One important conclusion of these considerations is that the spectrum of $H$ is absolutely continuous away of the ``Dirichlet points''. This is a consequence of the following claim.
 % -------------- %
\begin{proposition} \label{p: nonconst}
The solution of the equation $\det(M)=0$ regarded as a function of $(\theta_1,\theta_2)$ is non-constant on any open subset of $(-\pi,\pi]^2$.
\end{proposition}
 % -------------- %
\begin{proof}
Let us denote $F(\theta_1,\theta_2,k) =-\frac{\e^{\i(\theta_1 +\theta_2)}}{4}\det(M)$ and consider first the positive-energy solutions, i.e., values $k>0$ satisfying the condition $\sin ak\neq 0$. Obviously, a number $k$ solves $\det(M)=0$ for a pair $(\theta_1,\theta_2)\in(-\pi,\pi]^2$ if and only if $F(\theta_1,\theta_2,k)=0$. We use a \emph{reductio ad absurdum} argument. Suppose that the function $k=k(\theta_1,\theta_2)$ is constant on an open subset $J\subset(-\pi,\pi]^2$, i.e., let $F(\theta_1,\theta_2,k_0)=0$ hold for a $k_0>0$ and for every $(\theta_1,\theta_2) \in J$. Hence in view of~\eqref{det(M)} and the definition of $F$ we have
 % -------------- %
$$
\sin ak_0\cos(\theta_1-\theta_2)+\sin ck_0\cos\theta_1+\sin bk_0\cos\theta_2=\mathrm{const} \quad\text{on $J$}.
$$
 % -------------- %
The trigonometric polynomial $A\cos(\theta_1-\theta_2)+C\cos\theta_1 +B\cos\theta_2$ regarded as a function of two variables $(\theta_1,\theta_2)$ can be obviously constant on a non-empty open subset of $(-\pi,\pi]^2$ if and only if $A=B=C=0$ which in our case would mean $\sin ak_0=\sin bk_0=\sin ck_0=0$, however, this is excluded by the assumption. \\
In case of negative energies $-\kappa^2$ with $\kappa>0$ we have instead a condition
$$
\sinh a\kappa_0\cos(\theta_1-\theta_2)+\sinh c\kappa_0\cos\theta_1+\sinh b\kappa_0\cos\theta_2=\mathrm{const} \quad\text{on $J$},
$$
which can never be satisfied for $\kappa>0$.
\end{proof}
 % -------------- %

\noindent At the same time, the above argument allows us to \emph{prove Proposition~\ref{Bod nutna}}. Indeed, in view of the periodicity the point spectrum has necessarily an infinite multiplicity, corresponding to a ``flat band'', i.e. a solution to the spectral condition independent of $(\theta_1,\theta_2)$.
We have seen in Proposition~\ref{p: nonconst} that this can happen only if the energy is positive.
We can also exclude the case when all the edge lengths are commensurate as we already know that $\frac{b}{a}\in\Q$ and $\frac{c}{a}\in\Q$ implies $\sigma_\mathrm{p}(H)\ne\emptyset$.
Let $k^2>0$ and at least two of the lengths be incommensurate. Then $\sin ak$, $\sin bk$, and $\sin ck$ cannot vanish simultaneously. We choose a nonzero one and if needed renumber the edges in order to satisfy the assumption $\sin ak\ne 0$. Then Proposition~\ref{p: nonconst} implies that $k$ cannot corrrespond to a ``flat band''. \hfill $\Box$

 % -------------- %
\begin{corollary}
If $\frac{b}{a}\not\in\Q$ and $\frac{c}{a}\not\in\Q$, the spectrum of $H$ is purely absolutely continuous.
\end{corollary}
 % -------------- %
\begin{proof}
By Proposition~\ref{Bod nutna} the spectrum is purely continuous. By implicit-function theorem any solution to the conditions (\ref{SP theta}) is smooth, even analytic, hence singularly continuous spectrum is excluded.
\end{proof}
 % -------------- %

\noindent Let us add that if the edge lengths are commensurate, the operator may have infinitely degenerate eigenvalues, however, the implicit-function-theorem argument still works and the spectrum is absolutely continuous away from the ``Dirichlet points''.

\subsection{More about the spectral condition for $E=k^2>0$}\label{Sect. positive}

Consider again the positive part of the spectrum and examine the range of the right-hand side of~\eqref{SP theta} for $\theta_1,\theta_2 \in(-\pi,\pi]$. The range is obviously an interval. The maximum is found easily; using
 % -------------- %
$$
\frac{\cos\theta_1}{\sin ak\sin bk} \le \frac{1}{|\sin ak\sin bk|}
$$
 % -------------- %
and similar estimates for the other two $\theta$-dependent terms, we get
 % -------------- %
\begin{multline*}
\max_{\theta_1,\theta_2\in(-\pi,\pi]}\bigg\{\frac{1}{\sin^2 ak}+\frac{1}{\sin^2 bk}+\frac{1}{\sin^2 ck}+2\bigg(\frac{\cos\theta_1}{\sin ak\sin bk} \\ +\frac{\cos\theta_2}{\sin ak\sin ck} +\frac{\cos(\theta_1-\theta_2)}{\sin bk\sin ck}\bigg)\bigg\}
=\left(\frac{1}{|\sin ak|}+\frac{1}{|\sin bk|}+\frac{1}{|\sin ck|}\right)^2\,.
\end{multline*}
 % -------------- %
The maximum is obviously attained for $\theta_1,\theta_2$ chosen such that $\cos\theta_1=\sgn(\sin ak\sin bk)$, $\cos\theta_2=\sgn(\sin ak\sin ck)$. On the other hand, the minimum of the expression will be found using the following lemma which is not difficult to prove.

 % -------------- %
\begin{lemma}\label{Lem. 3.3}
Let $f(\theta_1,\theta_2)=A\cos(\theta_1-\theta_2)+B\cos\theta_2+C\cos\theta_1$ for $A,B,C\in\R$ such that $ABC>0$. It holds
 % -------------- %
\begin{itemize}
 % -------------- %
\item if $\frac{1}{|A|}+\frac{1}{|B|}+\frac{1}{|C|}\geq2\max\left\{\frac{1}{|A|},\frac{1}{|B|},\frac{1}{|C|}\right\}$, then
 % -------------- %
$$
\min_{\theta_1,\theta_2\in(-\pi,\pi]} f(\theta_1,\theta_2)=-\frac{ABC}{2}\left(\frac{1}{A^2}+\frac{1}{B^2}+\frac{1}{C^2}\right)\,;
$$
 % -------------- %
\item if $\frac{1}{|A|}+\frac{1}{|B|}+\frac{1}{|C|}\leq2\max\left\{\frac{1}{|A|},\frac{1}{|B|},\frac{1}{|C|}\right\}$, then
 % -------------- %
$$
\min_{\theta_1,\theta_2\in(-\pi,\pi]} f(\theta_1,\theta_2)=-(|A|+|B|+|C|)+2\min\{|A|,|B|,|C|\}\,.
$$
 % -------------- %
\end{itemize}
\end{lemma}
 % -------------- %

\noindent Let us apply the result on the right-hand side of~\eqref{SP theta}. We need to set $A=\sin bk\sin ck$, $B=\sin ak\sin ck$, $C=\sin ak\sin bk$. Then the condition $\frac{1}{|A|}+\frac{1}{|B|}+\frac{1}{|C|}\geq2\max\left\{\frac{1}{|A|},\frac{1}{|B|},\frac{1}{|C|}\right\}$ can be shown to be equivalent to $\frac{1}{|\sin ak|}+\frac{1}{|\sin bk|}+\frac{1}{|\sin ck|}\geq2\max\left\{\frac{1}{|\sin ak|},\frac{1}{|\sin bk|},\frac{1}{|\sin ck|}\right\}$ (and similarly for the opposite sign). When we substitute the minima of $f$ found in Lemma~\ref{Lem. 3.3} into the right-hand side of~\eqref{SP theta}, we get

 % -------------- %
\begin{itemize}
 % -------------- %
\item zero if $\frac{1}{|\sin ak|}+\frac{1}{|\sin bk|}+\frac{1}{|\sin ck|}\geq2\max\left\{\frac{1}{|\sin ak|},\frac{1}{|\sin bk|},\frac{1}{|\sin ck|}\right\}$;
 % -------------- %
\item $\left(2\max\left\{\frac{1}{|\sin ak|},\frac{1}{|\sin bk|},\frac{1}{|\sin ck|}\right\}-\frac{1}{|\sin ak|}-\frac{1}{|\sin bk|}-\frac{1}{|\sin ck|}\right)^2$ otherwise.
     % -------------- %
\end{itemize}
 % -------------- %
The results on the minimum and maximum allow us to estimate the left-hand side of~\eqref{SP theta} from below and above; taking the square roots we get the condition
 % -------------- %
\begin{eqnarray*}
\lefteqn{\max\bigg\{0,2\max\left\{\frac{1}{|\sin ak|},\frac{1}{|\sin bk|},\frac{1}{|\sin ck|}\right\}-\bigg(\frac{1}{|\sin ak|}+\frac{1}{|\sin bk|}} \\ && \hspace{3em} +\frac{1}{|\sin ck|}\bigg)\bigg\}
\leq\left|\cotg ak+\cotg bk+\cotg ck+\frac{\alpha}{k}\right| \\ &&
\hspace{8em} \leq\frac{1}{|\sin ak|}+\frac{1}{|\sin bk|}+\frac{1}{|\sin ck|}\,. \phantom{AAAAAAAA}
\end{eqnarray*}
 % -------------- %
The first term at the left-hand side is obviously non-negative, hence we arrive at the conclusion which can be stated as two gap conditions:
 % -------------- %
\begin{itemize}
 % -------------- %
\item Condition GC1: $\sigma(H)$ has a gap if
 % -------------- %
\begin{equation}\label{shora gap}
\left|\cotg ak+\cotg bk+\cotg ck+\frac{\alpha}{k}\right|>\frac{1}{|\sin ak|}+\frac{1}{|\sin bk|}+\frac{1}{|\sin ck|}\,;
\end{equation}
 % -------------- %
\item Condition GC2: $\sigma(H)$ has a gap if
 % -------------- %
\begin{multline}\label{zdola gap}
2\max\left\{\frac{1}{|\sin ak|},\frac{1}{|\sin bk|},\frac{1}{|\sin ck|}\right\}-\left(\frac{1}{|\sin ak|}+\frac{1}{|\sin bk|}+\frac{1}{|\sin ck|}\right) \\
>\left|\cotg ak+\cotg bk+\cotg ck+\frac{\alpha}{k}\right|\,.
\end{multline}
 % -------------- %
\end{itemize}
 % -------------- %
We will consider them separately.

\subsection{Gap condition GC1}

 % -------------- %
\begin{observation}\label{stejna znamenka}
If the gap condition GC1 \eqref{shora gap} is satisfied, then
 % -------------- %
$$
\sgn(\cotg ak)=\sgn(\cotg bk)=\sgn(\cotg ck)=\sgn(\alpha) \quad\vee\quad k<|\alpha|\,.
$$
 % -------------- %
\end{observation}
 % -------------- %
\begin{proof}
We employ \emph{reductio ad absurdum}. Let $k\geq|\alpha|$ and, for instance, $\sgn(\cotg ak)=-\sgn(\alpha)$. We have
 % -------------- %
$$
\left|\cotg ak+\cotg bk+\cotg ck+\frac{\alpha}{k}\right|\leq|\cotg bk|+|\cotg ck|+\left|\cotg ak+\frac{\alpha}{k}\right|\,.
$$
 % -------------- %
Since $\cotg ak$ and $\alpha$ have opposite signs and $|\cotg x|\leq\frac{1}{|\sin x|}$ for any admissible $x\in\R$, it holds
 % -------------- %
$$
\left|\cotg ak+\frac{\alpha}{k}\right|\leq \max\left\{|\cotg ak|,\frac{|\alpha|}{k}\right\}\leq\max\left\{|\cotg ak|,1\right\}\leq\frac{1}{|\sin ak|}\,.
$$
 % -------------- %
Hence
 % -------------- %
$$
\left|\cotg ak+\cotg bk+\cotg ck+\frac{\alpha}{k}\right|\leq\frac{1}{|\sin ak|}+\frac{1}{|\sin bk|}+\frac{1}{|\sin ck|}\,,
$$
 % -------------- %
i.e., the gap condition GC1~\eqref{shora gap} is violated.
\end{proof}
 % -------------- %

Let $\|\cdot\|$ be the nearest-integer function on $\R$, that is, $\|x\|$ is the nearest integer to $x\in\R$. In the following we will need the function the value of which represents the difference between a given number and the nearest integer, i.e. $x\mapsto x-\|x\|$. For the sake of brevity, we introduce the symbol
 % -------------- %
\begin{equation}\label{symbol}
\{x\}:=x-\|x\|\,;
\end{equation}
 % -------------- %
it holds obviously $\{x\}\in[-1/2,1/2]$ for any $x\in\R$\,.

 % -------------- %
\begin{corollary}\label{gap podm}
For $k\geq|\alpha|$, the gap condition~\eqref{shora gap} is satisfied if and only if $\sgn(\cotg ak)=\sgn(\cotg bk)=\sgn(\cotg ck)=\sgn(\alpha)$ and
 % -------------- %
$$
\left|\tg\left(\left\{\frac{ak}{\pi}\right\}\frac{\pi}{2}\right)\right|
+\left|\tg\left(\left\{\frac{bk}{\pi}\right\}\frac{\pi}{2}\right)\right|
+\left|\tg\left(\left\{\frac{ck}{\pi}\right\}\frac{\pi}{2}\right)\right|<\frac{|\alpha|}{k}\,,
$$
 % -------------- %
where $\{\cdot\}$ is the function defined by~\eqref{symbol}.
\end{corollary}
 % -------------- %
\begin{proof}
Suppose that $k\geq|\alpha|$ and~\eqref{shora gap} holds. It follows from Observation~\ref{stejna znamenka} that condition~\eqref{shora gap} implies $\sgn(\cotg ak)=\sgn(\cotg bk)=\sgn(\cotg ck)=\sgn(\alpha)$. The inequality~\eqref{shora gap} is thus equivalent to $\sgn(\cotg ak)=\sgn(\cotg bk)=\sgn(\cotg ck)=\sgn(\alpha)$ together with
 % -------------- %
\begin{equation}\label{opacna 2}
|\cotg ak|+|\cotg bk|+|\cotg ck|+\frac{|\alpha|}{k} >\frac{1}{|\sin ak|}+\frac{1}{|\sin bk|}+\frac{1}{|\sin ck|}\,.
\end{equation}
 % -------------- %
For any $x\in\R$, it holds
 % -------------- %
\begin{multline*}
\frac{1}{|\sin x|}-|\cotg x|=\frac{1-|\cos x|}{|\sin x|}=\left\{\begin{array}{ll}
\frac{1-\cos x}{|\sin x|}=\left|\tg\frac{x}{2}\right| & \text{for } \cos x>0 \\[.5em]
\frac{1+\cos x}{|\sin x|}=\left|\cotg\frac{x}{2}\right| & \text{for } \cos x<0
\end{array}\right\}
\\
=\min\left\{\left|\tg\frac{x}{2}\right|,\left|\cotg\frac{x}{2}\right|\right\}
=\left|\tg\left(\left\{\frac{x}{2}\cdot\frac{2}{\pi}\right\}\frac{\pi}{2}\right)\right|
=\left|\tg\left(\left\{\frac{x}{\pi}\right\}\frac{\pi}{2}\right)\right|\,.
\end{multline*}
 % -------------- %
Consequently, \eqref{opacna 2} can be rewritten as
 % -------------- %
$$
\left|\tg\left(\left\{\frac{ak}{\pi}\right\}\frac{\pi}{2}\right)\right|
+\left|\tg\left(\left\{\frac{bk}{\pi}\right\}\frac{\pi}{2}\right)\right|
+\left|\tg\left(\left\{\frac{ck}{\pi}\right\}\frac{\pi}{2}\right)\right|<\frac{|\alpha|}{k}\,.
$$
 % -------------- %
To conclude, the last inequality together with the condition $\sgn(\cotg ak)=\sgn(\cotg bk)=\sgn(\cotg ck)=\sgn(\alpha)$ is equivalent to gap condition~\eqref{shora gap}, as we have set up to prove.
\end{proof}
 % -------------- %

 % -------------- %
\begin{observation}
Local minima of the function
 % -------------- %
$$
F(k):=\left|\tg\left(\left\{\frac{ak}{\pi}\right\}\frac{\pi}{2}\right)\right|+\left|\tg\left(\left\{\frac{bk}{\pi}\right\}\frac{\pi}{2}\right)\right|+\left|\tg\left(\left\{\frac{ck}{\pi}\right\}\frac{\pi}{2}\right)\right|
$$
 % -------------- %
for $k>0$ occur at the points $\frac{m\pi}{a}$, $\frac{m\pi}{b}$, $\frac{m\pi}{c}$ with $m\in\N$.
\end{observation}
 % -------------- %

\noindent An immediate consequence, in combination with Corollary~\ref{gap podm}, is that the spectrum has open gaps for any $\alpha\ne 0$ when the lattice edges are commensurate. If at least two of them are not commensurate, existence of gaps due the condition GC1 depend on how fast the minima of $F(k)$ decrease as $k\to\infty$; we will discuss it in more detail in the next section.

\subsection{Gap condition GC2}

Obviously, condition GC2 can be satisfied only if $\frac{1}{|\sin\ell_1 k|}>\frac{1}{|\sin\ell_2 k|}+\frac{1}{|\sin\ell_3 k|}$ holds for a certain choice $\{\ell_1,\ell_2,\ell_3\}=\{a,b,c\}$. We begin with the following auxiliary result.

 % -------------- %
\begin{lemma}\label{Lem. trigon}
If $x_1,x_2,\ldots,x_N$ are all greater or equal to one and they satisfy $x_1>x_2+\cdots+x_N$, then
 % -------------- %
$$
x_1-\sum_{i=2}^N x_i < \sqrt{x_1^2-1}-\sum_{i=2}^N \sqrt{x_i^2-1}\,.
$$
 % -------------- %
\end{lemma}
 % -------------- %
\begin{proof}
We prove the statement by induction in $N$. To begin with, we prove for $N=2$ and any $x_1>x_2\geq1$ the implication
 % -------------- %
$$
x_1>x_2 \quad\Rightarrow\quad x_1-x_2<\sqrt{x_1^2-1}-\sqrt{x_2^2-1}\,.
$$
 % -------------- %
We rewrite this statement as $x_1-\sqrt{x_1^2-1}<x_2-\sqrt{x_2^2-1}$, which is equivalent to
 % -------------- %
$$
\frac{1}{x_1+\sqrt{x_1^2-1}}<\frac{1}{x_2+\sqrt{x_2^2-1}}\,,
$$
 % -------------- %
and this is obviously valid under the assumption $x_1>x_2$. Next we assume that the claim holds true for an $N\geq 2$, and we want to demonstrate for any $x_1,x_2,\ldots,x_N,x_{N+1}\geq1$ the implication
 % -------------- %
$$
x_1>x_2+\cdots+x_{N+1} \quad\Rightarrow\quad x_1-\sum_{i=2}^{N+1} x_i < \sqrt{x_1^2-1}-\sum_{i=2}^{N+1}\sqrt{x_i^2-1}\,.
$$
 % -------------- %
We set $x_N+x_{N+1}=y$. The induction hypothesis applied on the $N$-tuple $x_1,\ldots,x_{N-1},y$ implies that $x_1-x_2-\cdots -x_{N-1}-(x_N+x_{N+1})$ is less than
 % -------------- %
$$
\sqrt{x_1^2-1}-\sum_{i=2}^{N-1} \sqrt{x_i^2-1} - \sqrt{(x_N+x_{N+1})^2-1}\,;
$$
 % -------------- %
thus it suffices to check for any $x_N,x_{N+1}\geq1$ the inequality
 % -------------- %
$$
\sqrt{(x_N+x_{N+1})^2-1}>\sqrt{x_N^2-1}+\sqrt{x_{N+1}^2-1}\,,
$$
 % -------------- %
which is a straightforward task.
\end{proof}

 % -------------- %
\begin{corollary}\label{Cor. trigon}
If $\frac{1}{|\sin\ell_1 k|}>\frac{1}{|\sin\ell_2 k|}+\frac{1}{|\sin\ell_3 k|}$ and $\alpha\cotg\ell_1 k\geq0$, then
 % -------------- %
$$
\frac{1}{|\sin\ell_1 k|}-\frac{1}{|\sin\ell_2 k|}-\frac{1}{|\sin\ell_3 k|}\leq\left|\cotg\ell_1 k+\cotg\ell_2 k+\cotg\ell_3+\frac{\alpha}{k}\right|\,.
$$
 % -------------- %
\end{corollary}
 % -------------- %
\begin{proof}
In view of the assumption $\alpha\cotg\ell_1 k\geq0$ we have
 % -------------- %
\begin{eqnarray*}
\lefteqn{\left|\cotg\ell_1 k+\cotg\ell_2 k+\cotg\ell_3+\frac{\alpha}{k}\right|\geq\left|\cotg\ell_1 k+\frac{\alpha}{k}\right|-|\cotg\ell_2 k|} \\ && \hspace{3em} -|\cotg\ell_3 k| \geq|\cotg\ell_1 k|-|\cotg\ell_2 k|-|\cotg\ell_3|\,;
\end{eqnarray*}
 % -------------- %
thus it suffices to prove
 % -------------- %
$$
\frac{1}{|\sin\ell_1 k|}-\frac{1}{|\sin\ell_2 k|}-\frac{1}{|\sin\ell_3 k|}\leq|\cotg\ell_1 k|-|\cotg\ell_2 k|-|\cotg\ell_3 k|\,.
$$
 % -------------- %
This is, however, a straightforward consequence of Lemma~\ref{Lem. trigon}, it is enough to set $N=3$, $x_1=\frac{1}{|\sin\ell_1 k|}$, $x_2=\frac{1}{|\sin\ell_2 k|}$ and $x_3=\frac{1}{|\sin\ell_3 k|}$.
\end{proof}

\noindent To sum up, the condition GC2 can give rise to an open gap only if the greatest element of the set $\left\{\frac{1}{|\sin ak|},\frac{1}{|\sin bk|},\frac{1}{|\sin ck|}\right\}$ is greater than the sum of the other two and the sign of the corresponding cotangent is opposite to the sign of $\alpha$. In particular, condition GC2 gives rise to no open gaps in the Kirchhoff case, $\alpha=0$.

\subsection{Negative spectrum}

Let us finally discuss briefly the negative spectrum of $H$, which is obviously nonempty if and only if $\alpha<0$. Spectral condition~\eqref{SP theta neg} can be rephrased into two gap conditions, similarly as it has been done in Section~\ref{Sect. positive} for $E>0$. Specifically, the gap conditions for $E=-\kappa^2$ acquire the following form:
 % -------------- %
\begin{equation}\label{GC1 neg}
\left|\coth a\kappa+\coth b\kappa+\coth c\kappa+\frac{\alpha}{\kappa}\right|>\frac{1}{\sinh a\kappa}+\frac{1}{\sin b\kappa}+\frac{1}{\sin c\kappa}\,,
\end{equation}
 % -------------- %
\begin{equation}\label{GC2 neg}
\left|\coth a\kappa+\coth b\kappa+\coth c\kappa+\frac{\alpha}{\kappa}\right|<\frac{2}{\sinh\ell_{\min}\kappa}-\frac{1}{\sinh a\kappa}-\frac{1}{\sinh b\kappa}-\frac{1}{\sinh c\kappa}\,,
\end{equation}
 % -------------- %
where $\ell_{\min}:=\min\{a,b,c\}$. One can describe circumstances under which the spectrum has an open gap in its negative part.

 % -------------- %
\begin{proposition}
The negative part of $\sigma(H)$ contains a gap adjacent to zero exactly in the following two cases:
 % -------------- %
\begin{itemize}
\setlength{\itemsep}{5pt}
 % -------------- %
\item $|\alpha|>\frac{2}{a}+\frac{2}{b}+\frac{2}{c}$,
 % -------------- %
\item $\frac{2}{\ell_{\min}}>\frac{1}{a}+\frac{1}{b}+\frac{1}{c}\,$ and $\,\frac{2}{a}+\frac{2}{b}+\frac{2}{c}-\frac{2}{\ell_{\min}}<|\alpha|<\frac{2}{\ell_{\min}}$.
 % -------------- %
\end{itemize}
 % -------------- %
\end{proposition}

\begin{proof}
We begin with condition~\eqref{GC1 neg} and compare the asymptotic behavior of the two sides in the limit $\kappa\searrow0$. Up to higher-order term we have
 % -------------- %
$$
\left|\coth a\kappa+\coth b\kappa+\coth c\kappa-\frac{|\alpha|}{\kappa}\right|\approx
\frac{1}{\kappa}\left|\frac{1}{a}+\frac{1}{b}+\frac{1}{c}-|\alpha|\right|\,,
$$
 % -------------- %
$$
\frac{1}{|\sinh a\kappa|}+\frac{1}{|\sin b\kappa|}+\frac{1}{|\sin c\kappa|}\approx\frac{1}{\kappa}\left(\frac{1}{a}+\frac{1}{b}+\frac{1}{c}\right)\,,
$$
 % -------------- %
hence the first gap condition can be for small values of $\kappa$ satisfied provided $\left|\frac{1}{a}+\frac{1}{b} +\frac{1}{c}-|\alpha|\right|> \frac{1}{a}+\frac{1}{b} +\frac{1}{c}$, which is true if and only if $|\alpha|>\frac{2}{a}+\frac{2}{b}+\frac{2}{c}$. Let us proceed to \eqref{GC2 neg}. In the regime $\kappa\searrow0$ we have
 % -------------- %
$$
\frac{2}{\sinh\ell_{\min}\kappa}-\frac{1}{\sinh a\kappa}-\frac{1}{\sinh b\kappa}-\frac{1}{\sinh c\kappa}\approx\frac{1}{\kappa}\left(\frac{2}{\ell_{\min}}
-\frac{1}{a}-\frac{1}{b}-\frac{1}{c}\right)\,,
$$
 % -------------- %
therefore the condition acquires for small values of $\kappa$ the form
 % -------------- %
$$
\left|\frac{1}{a}+\frac{1}{b}+\frac{1}{c}-|\alpha|\right|
<\frac{2}{\ell_{\min}}-\frac{1}{a}-\frac{1}{b}-\frac{1}{c}\,.
$$
 % -------------- %
This inequality can be satisfied only if $\frac{2}{\ell_{\min}} >\frac{1}{a}+\frac{1}{b}+\frac{1}{c}$, and under this condition it is valid \emph{iff} $\frac{2}{a}+\frac{2}{b}+\frac{2}{c} -\frac{2}{\ell_{\min}}<|\alpha|<\frac{2}{\ell_{\min}}$.
\end{proof}

\section{The case $b=c$} \label{s: b=c}

The spectral picture with respect to all four parameters of the model is rather complex. In order to simplify the discussion, we focus from now on at the case when the lattice can be stretched in one direction only assuming $b=c$. The above gap conditions acquire then the following form,

 % -------------- %
\begin{itemize}
 % -------------- %
\item Condition GC1: $\sigma(H)$ has a gap if
 % -------------- %
\begin{equation}\label{GC1 b=c}
\left|\cotg ak+2\cotg bk+\frac{\alpha}{k}\right|>\frac{1}{|\sin ak|}+\frac{2}{|\sin bk|}\,;
\end{equation}
 % -------------- %
\item Condition GC2: $\sigma(H)$ has a gap if
 % -------------- %
\begin{equation}\label{GC2 b=c}
\frac{1}{|\sin ak|}-\frac{2}{|\sin bk|}>\left|\cotg ak+2\cotg bk+\frac{\alpha}{k}\right|\,;
\end{equation}
 % -------------- %
note that \eqref{GC2 b=c} cannot be satisfied if $|\sin bk|<|\sin ak|$.
\end{itemize}
 % -------------- %

\subsection{Gap condition GC1 for $b=c$}

According to Corollary~\ref{gap podm}, the gap condition is equivalent to the conditions $\sgn(\cotg ak)=\sgn(\cotg bk)=\sgn(\alpha)$ and $F(k)<\frac{|\alpha|}{k}$, where
 % -------------- %
$$
F(k)=\left|\tg\left(\left\{\frac{ak}{\pi}\right\}\frac{\pi}{2}\right)\right|
+2\left|\tg\left(\left\{\frac{bk}{\pi}\right\}\frac{\pi}{2}\right)\right|
$$
 % -------------- %
with $\{\cdot\}$ defined by~\eqref{symbol}. To state the next result, we have to introduce two classes of irrational numbers. A $\theta\in\R$ is called \emph{badly approximable} if there exists a constant $\gamma>0$ such that $\left|\theta -\frac{p}{q}\right|>\frac{\gamma}{q^2}$ holds for all $p,q\in\N$. Irrational numbers that do not have this property will be called, following \cite{Ex96}, \emph{Last admissible} \cite{La94}. Thus a $\theta\in\R\backslash\Q$ is Last admissible if there exist increasing integer sequences $\left\{p_n\right\}_{n=1}^\infty$, $\left\{q_n\right\}_{n=1}^\infty$ such that $\lim_{n\to\infty} q_n^2\left|\theta-\frac{p_n}{q_n}\right|=0$. Another way to characterize them is through the continued fraction representation: a number $\theta=[a_0;a_1,a_2,\dots]$ belongs to the class of Last admissible numbers if the coefficient sequence $\{a_j\}$ is unbounded.

 % -------------- %
\begin{theorem}\label{Proposition GC1}
Let $\theta=\frac{a}{b}$.
 % -------------- %
\begin{itemize}
 % -------------- %
\item[(i)] If $\theta\in\Q$, then the gap condition GC1 generates infinitely many gaps in the spectrum of $H$ for any $\alpha\neq0$.
 % -------------- %
\item[(ii)] If $\theta$ in a Last admissible irrational number, then the gap condition GC1 generates infinitely many gaps for any $\alpha\neq0$.
 % -------------- %
\item[(iii)] Let $\theta$ be a badly approximable irrational number. There is a positive $\alpha_0$ such that the condition GC1 generates no gaps provided $0\le |\alpha|\le\alpha_0$. On the other hand, if the coupling constant satisfies $|\alpha|> \frac{4\pi}{\sqrt{5}}\min\{\frac{2}{a},\frac{1}{b}\}$, there are infinitely many gaps.
\end{itemize}
 % -------------- %
\end{theorem}
 % -------------- %
\begin{proof}
(i) If $\theta$ is rational, then there are obviously infinitely many positive integers $m$ such that $am,bm$ are even numbers, and therefore, $F(k)=0$ holds for $k=m\pi$. Moreover, if $k=m\pi+\delta$ for a sufficiently small $\delta>0$, it holds $F(k+\sgn(\alpha)\cdot\delta)<\frac{|\alpha|}{k}$ and $\sgn(\cotg ak)=\sgn(\cotg bk)=\sgn(\alpha)$. Corollary~\ref{gap podm} then implies the existence of infinitely many gaps.

(ii) If $\theta$ is Last admissible, then there exist increasing integer sequences $\left\{p_n\right\}_{n=1}^\infty$, $\left\{q_n\right\}_{n=1}^\infty$ such that $\lim_{n\to\infty} q_n^2\left|\theta-\frac{p_n}{q_n}\right|=0$. Moreover, one can find sequences having, in addition, the property $\theta-\frac{p_n}{q_n}>0$ or $\theta-\frac{p_n}{q_n}<0$, respectively. Let us choose the sequences such that $\sgn\left(\theta -\frac{p_n}{q_n}\right)=\sgn(\alpha)$, and set $k_n=\frac{q_n\pi}{b}$. Obviously, $\sgn(\cotg(a k_n))=\sgn(\left\{\frac{a}{b}q_n\right\})$. It holds
\begin{equation}\label{aq/b}
\left\{\frac{a}{b}q_n\right\}=\theta q_n-\|\theta q_n\|=q_n\left(\theta-\frac{\|\theta q_n\|}{q_n}\right)=q_n\left(\theta-\frac{p_n}{q_n}\right)\,,
\end{equation}
where we have used the fact that $p_n$ is equal to $\|\theta q_n\|$, which immediately follows from $\lim_{n\to\infty}q_n^2\left|\theta -\frac{p_n}{q_n}\right|=0$.
Consequently, the equality $\sgn(\cotg(a k))=\sgn(\alpha)$ holds for $k=k_n$, as well as for $k$ in a certain neighbourhood of $k_n$. Furthermore, \eqref{aq/b} implies that
% -------------- %
\begin{multline*}
q_n\lim_{k\to\frac{q_n\pi}{b}}F(k)
=q_n\left|\tg\left(\left\{\frac{a}{b}q_n\right\}\frac{\pi}{2}\right)\right|
<4q_n\left|\left\{\frac{a}{b}q_n\right\}\right|=4q_n^2\left|\theta-\frac{p_n}{q_n}\right|\to 0\
\end{multline*}
 % -------------- %
holds as $n\to\infty$. At the same time,
 % -------------- %
$$
q_n\lim_{k\to\frac{q_n\pi}{b}}\frac{|\alpha|}{k}=\frac{|\alpha|b}{\pi}>0 \qquad\text{for all $n\in\N$}\,.
$$
 % -------------- %
Comparing the two limits, we see that for any $n\in\N$ there exists a neighbourhood of $k_n$ on which it holds $F(k)<\frac{|\alpha|}{k}$. If we choose the \emph{right} neighbourhood for $\alpha>0$ and the \emph{left} neighbourhood for $\alpha<0$, the remaining condition $\sgn(\cotg bk)=\sgn(\alpha)$ will be satisfied there as well. To sum up, we have found infintely many points $k_n$ with certain neighbourhoods on which the gap condition~\eqref{GC1 b=c} is satisfied. In other words, the spectrum of the Hamiltonian has infinitely many gaps located at certain integer multiples of $\frac{\pi^2}{b^2}$. In the same way one can check the existence of neighbourhoods of a sequence of points $\frac{q_n\pi}{a}$ where the gap condition is satisfied.

(iii) When $\theta$ is badly approximable, there exists, by definition, a constant $\gamma>0$ such that $\left|\theta -\frac{p}{q}\right|>\frac{\gamma}{q^2}$ holds for all $p,q\in\N$.
This yields
 % -------------- %
\begin{eqnarray*}
\lefteqn{F\left(\frac{m\pi}{b}\right)
=\left|\tg\left(\left\{\frac{a}{b}m\right\}\frac{\pi}{2}\right)\right|
>\left|\left\{\frac{a}{b}m\right\}\right|\frac{\pi}{2}=\left|\theta m-\|\theta m\|\right|\frac{\pi}{2}} \\ && \hspace{5em}=m\left|\theta-\frac{\|\theta m\|}{m}\right|\frac{\pi}{2}>m\frac{\gamma}{m^2}\frac{\pi}{2}
=\frac{\gamma\pi}{2m}\,, \phantom{AAAAAA}
\end{eqnarray*}
 % -------------- %
and consequently, $F\left(\frac{m\pi}{b}\right)> |\alpha|\left( \frac{m\pi}{b}\right)^{-1}$ holds if $|\alpha|\leq\frac{\gamma\pi^2}{2b}$, i.e. the condition~\eqref{GC1 b=c} is violated in this case in all the local minima $\frac{m\pi}{b}$ of $F$.

It remains to show that the gap condition is violated in the local minima of $F$ at the points $\frac{m\pi}{a}$ as well. It is a well known fact that a number $\theta$ is badly approximable if and only if $1/\theta$ is badly approximable. Moreover, if $\gamma>0$ is the minimal constant such that $\left|\theta-\frac{p}{q}\right|>\frac{\gamma}{q^2}$ holds for all $p,q\in\N$, then $\gamma$ is at the same time the minimal constant such that $\left|\frac{1}{\theta}-\frac{q}{p}\right|>\frac{\gamma}{p^2}$ for all $q,p\in\N$. Hence we obtain, similarly as above,
 % -------------- %
$$
F\left(\frac{m\pi}{a}\right)
=2\left|\tg\left(\left\{\frac{b}{a}m\right\}\frac{\pi}{2}\right)\right|
>2\left|\left\{\frac{b}{a}m\right\}\right|\frac{\pi}{2}
>2m\frac{\gamma}{m^2}\frac{\pi}{2}=\frac{\gamma\pi}{m}\,.
$$
 % -------------- %
Thus $F\left(\frac{m\pi}{a}\right)>|\alpha|\left(\frac{m\pi}{a} \right)^{-1}$ holds if $|\alpha|\leq\frac{\gamma\pi^2}{a}$, i.e. the gap condition GC1 is violated at the local minima $\frac{m\pi}{a}$ of $F$ as well. To sum up, for any $\alpha$ such that $0<|\alpha| <\gamma\pi^2\min\{ \frac{1}{a},\frac{1}{2b}\}$ all the local minima of $F$ satisfy $F(k)>\frac{|\alpha|}{k}$, in other words, the condition GC1~\eqref{GC1 b=c} is violated everywhere for $k>0$.

On the other hand, by the Hurwitz extension of the Dirichlet theorem \cite[Chap.~II]{Sch91} for any irrational $\theta$ there are increasing integer sequences $\left\{p_n\right\}_{n=1}^\infty$ and $\left\{q_n \right\}_{n=1}^\infty$ such that $\left|\theta-\frac{p_n}{q_n} \right|<\frac{1}{\sqrt{5}q_n^2}$ holds for all $n\in\N$. In addition, one can find such sequences with the property $\theta-\frac{p_n}{q_n}>0$ or $\theta-\frac{p_n}{q_n}<0$ for all $n\in\N$, respectively. This allows us to assume that $\sgn\left(\theta-\frac{p_n}{q_n}\right)=\sgn(\alpha)$, and setting $k_n:=\frac{q_n\pi}{b}$, we obtain
 % -------------- %
$$
F\left(k_n\right)=\left|\tg\left(\left\{\frac{a}{b}q_n\right\}\frac{\pi}{2}\right)\right|\,.
$$
 % -------------- %
Since $\{x\}\in[-1/2,1/2]$ holds for any $x\in\R$ by definition, we infer that $\left|\left\{\frac{a}{b}m\right\}\frac{\pi}{2}\right| \leq\frac{\pi}{4}$. Furthermore, since $|\tg x|\leq\frac{4}{\pi}|x|$ holds for any $|x|\leq\frac{\pi}{4}$, we get
 % -------------- %
$$
F\left(k_n\right)<2\left|\left\{\frac{a}{b}q_n\right\}\right|=2q_n\left|\theta-\frac{\|\theta q_n\|}{q_n}\right|=4q_n\left|\theta-\frac{p_n}{q_n}\right|<2q_n\frac{1}{\sqrt{5}q_n^2}=\frac{2}{\sqrt{5}q_n}\,.
$$
 % -------------- %
At the same time, we have
 % -------------- %
$$
\frac{|\alpha|}{k_n}=\frac{|\alpha|b}{q_n\pi}\,,
$$
 % -------------- %
and consequently, $|\alpha|>\frac{2\pi}{\sqrt{5}b}$ implies existence of neighbourhoods of $\frac{q_n\pi}{b}$ on which the gap condition is satisfied. In a similar way one can prove that for $|\alpha|>\frac{4\pi}{\sqrt{5}a}$ there are neighbourhoods of $\frac{q_n\pi}{a}$ on which the condition GC1 is satisfied. To conclude, the spectrum of $H$ has infinitely many open gaps generated by the condition~\eqref{GC1 b=c} provided $|\alpha|>\frac{2\pi}{\sqrt{5}} \min\{\frac{2}{a}, \frac{1}{b}\}$.
\end{proof}

\subsection{Gap condition GC2 for $b=c$}

As we have indicated, the ``lower'' gap condition acquires now the form~\eqref{GC2 b=c}.

 % -------------- %
\begin{lemma}\label{Lemma GC2}
If $k>|\alpha|$ and the condition~\eqref{GC2 b=c} is satisfied, then necessarily $\frac{1}{|\sin ak|}>\frac{2}{|\sin bk|}$, $\:\alpha\cotg ak<0$ and $\cotg ak\cotg bk<0$.
\end{lemma}
 % -------------- %
\begin{proof}
Inequality~\eqref{GC2 b=c} implies $\frac{1}{|\sin ak|}-\frac{2}{|\sin bk|}>0$, which gives the first claim. The second claim, $\alpha\cotg ak<0$, follows from Corollary~\ref{Cor. trigon}. It remains to show that $\cotg ak\cotg bk<0$. Note that $\frac{1}{|\sin ak|}>\frac{2}{|\sin bk|}$ implies $|\sin ak|<\frac{1}{2}$, hence $ak\in\left(m\pi-\frac{\pi}{6},m\pi+\frac{\pi}{6}\right)$ for an $m\in\N$, and therefore $|\cotg ak|>\sqrt{3}$.

We use again \emph{reductio ad absurdum} and suppose that $\cotg ak\cotg bk\geq0$. Then for any $k>|\alpha|$ we have, with regard to $|\cotg ak|>\sqrt{3}$,
 % -------------- %
$$
\left|\cotg ak+2\cotg bk+\frac{\alpha}{k}\right|\geq|\cotg ak|+2|\cotg bk|-\frac{|\alpha|}{k}\,,
$$
 % -------------- %
and since $-\frac{|\alpha|}{k}>-1$, we get
 % -------------- %
\begin{eqnarray*}
\lefteqn{\left|\cotg ak+2\cotg bk+\frac{\alpha}{k}\right|-\frac{1}{|\sin ak|}+\frac{2}{|\sin bk|}} \\ && >2\left(\frac{1}{|\sin bk|}+|\cotg bk|\right)-\left(\frac{1}{|\sin ak|}-|\cotg ak|\right)-1\,.
\end{eqnarray*}
 % -------------- %
It is easy to check that $\frac{1}{|\sin x|}+|\cotg x|\geq1$ and $\frac{1}{|\sin x|}-|\cotg x|\leq1$ for all $x\in\R$, hence
 % -------------- %
$$
\left|\cotg ak+2\cotg bk+\frac{\alpha}{k}\right|-\frac{1}{|\sin ak|}+\frac{2}{|\sin bk|}>0\,,
$$
 % -------------- %
which contradicts the inequality~\eqref{GC2 b=c}.
\end{proof}

 % -------------- %
\begin{corollary}\label{coro b=c}
For each $k>|\alpha|$ the gap condition~\eqref{GC2 b=c} is satisfied if and only if $\frac{1}{|\sin ak|}>\frac{2}{|\sin bk|}$, $\,\cotg ak\cotg bk<0$, $\,\alpha\cotg ak<0$, and $\left|G(k)-\frac{|\alpha|}{k}\right|<\frac{1}{|\sin ak|}-\frac{2}{|\sin bk|}$, where
\begin{equation}\label{G(k)}
G(k)=|\cotg ak|-2|\cotg bk|\,.
\end{equation}
 % -------------- %
\end{corollary}
 % -------------- %
\begin{proof}
With regard to Lemma~\ref{Lemma GC2}, the gap condition~\eqref{GC2 b=c} for a fixed $k>|\alpha|$ requires $\cotg ak\cotg bk<0$, $\,\alpha\cotg ak<0$. Consequently, the gap condition for $k>|\alpha|$ is satisfied if and only if $\cotg ak\cotg bk<0$, $\,\alpha\cotg ak<0$, and
\begin{equation}\label{GC2 uprava}
\left||\cotg ak|-2|\cotg bk|-\frac{|\alpha|}{k}\right|<\frac{1}{|\sin ak|}-\frac{2}{|\sin bk|}\,,
\end{equation}
which concludes the argument.
\end{proof}

Before we pass to analysis of the gaps generated by the condition GC2, we prove a lemma that will be useful in dealing with
rational ratio $\frac{a}{b}$ and with $\frac{a}{b}$ being a badly approximable irrational number.

 % -------------- %
\begin{lemma}\label{Lemma yyy}
Let $\theta=\frac{a}{b}$.
 % -------------- %
\begin{itemize}
 % -------------- %
\item[(i)] If $\theta\in\Q$, then there exists a $c>0$ such that for all $k>0$ it holds
 % -------------- %
$$
\cotg ak\cotg bk<0\,\wedge\,|\sin bk|\geq2|\sin ak| \quad\Rightarrow\quad G(k)\geq c\,.
$$
 % -------------- %
\item[(ii)] If $\theta$ is a badly approximable irrational number, then there exists a $c>0$ such that for all $k>0$ it holds
 % -------------- %
$$
\cotg ak\cotg bk<0\,\wedge\,|\sin bk|\geq2|\sin ak| \quad\Rightarrow\quad G(k)>\frac{c}{k}\,.
$$
\end{itemize}
 % -------------- %
\end{lemma}
 % -------------- %
\begin{proof}
Our aim is to estimate the function $G(k)=|\cotg ak|-2|\cotg bk|$ from below subject to the condition $\cotg ak\cotg bk<0\,\wedge\,|\sin bk|\geq2|\sin ak|$.

The function $G(k)$ attains local minima for $|\sin ak|=1$ and tends to $-\infty$ for $\sin bk=0$. Since both $|\sin ak|=1$ and $\sin bk=0$ contradict the condition $|\sin bk|\geq2|\sin ak|$, minimal values of $G(k)$ in the regions given by $\cotg ak\cotg bk<0\,\wedge\,|\sin bk|\geq2|\sin ak|$ are attained for $|\sin bk|=2|\sin ak|$.

The equality $|\sin bk|=2|\sin ak|$ gives
 % -------------- %
\begin{align*}
|\cotg ak|-2|\cotg bk|&=\frac{\sqrt{1-\sin^2 ak}}{|\sin ak|}-2\frac{\sqrt{1-\sin^2 bk}}{|\sin bk|} \\
&=\frac{\sqrt{1-\sin^2 ak}}{|\sin ak|}-\frac{\sqrt{1-4\sin^2 ak}}{|\sin ak|} \\
&=\frac{3|\sin ak|}{\sqrt{1-\sin^2 ak}+\sqrt{1-4\sin^2 ak}}\geq\frac{3|\sin ak|}{2}\,,
\end{align*}
 % -------------- %
hence $G(k)\geq\frac{3}{2}|\sin ak|$.

Let $m\in\N$ be chosen such that $m\pi$ is the integer multiple of $\pi$ closest to $ak$, i.e., $|ak-m\pi|\leq\frac{\pi}{2}$. In the same way we introduce $n\in\N$ such that $n\pi$ is the integer multiple of $\pi$ closest to $bk$. Obviously, the condition $\cotg ak\cdot\cotg bk<0$ implies $(ak-m\pi)\cdot(bk-n\pi)<0$.

It holds trivially $|\sin bk|\leq|bk-n\pi|$. The condition $|\sin bk|=2|\sin ak|$ implies $|\sin ak|\leq\frac{1}{2}$, hence $|ak-m\pi|\leq\frac{\pi}{6}$. Since $|x|\leq\frac{\pi}{6}\Rightarrow|\sin x|\geq\frac{3}{\pi}|x|$, we have $|\sin ak|=|\sin(ak-m\pi)|\geq\frac{3}{\pi}|ak-m\pi|$.

With regard to the estimates of $|\sin ak|$ and $|\sin bk|$ obtained above, it is easy to see that the quantity $|\sin ak|$ for $k$ solving the equation $|\sin bk|=2|\sin ak|$ is necessarily greater or equal to the quantity $\frac{3}{\pi}|ak-m\pi|$ for $k$ solving the equation $|bk-n\pi|=2\cdot\frac{3}{\pi}|ak-m\pi|$. Let us find such a $k$. The condition $(ak-m\pi)\cdot(bk-n\pi)<0$ together with $|bk-n\pi|=2\cdot\frac{3}{\pi}|ak-m\pi|$ gives the equation $bk-n\pi=-2\cdot\frac{3}{\pi}(ak-m\pi)$. Its solution reads $k'=\frac{6m+\pi n}{6a+\pi b}\pi$. Therefore, $\cotg ak\cotg bk<0\,\wedge\,|\sin bk|\geq2|\sin ak|$ implies
 % -------------- %
\begin{equation}\label{mi2}
G(k)\geq\frac{3}{2}\cdot\frac{3}{\pi}|ak'-m\pi|=\frac{9}{2\pi}\left|a\frac{6m+\pi n}{6a+\pi b}\pi-m\pi\right|=\frac{9\pi}{2}\cdot\frac{|an-bm|}{6a+\pi b}\,.
\end{equation}
 % -------------- %

\smallskip

\noindent (i) Let $\theta\in\Q$, i.e., $a=pL$, $b=qL$ for certain $p,q\in\N$ and $L>0$. Then the just obtained bound~\eqref{mi2} gives
 % -------------- %
$$
G(k)\geq\frac{9\pi}{2}\cdot\frac{\left|pLn-qLm\right|}{6pL+\pi qL}\,,
$$
 % -------------- %
where $L$ at the right-hand side can be obviously canceled. Note that the expression $pn-qm$ is necessarily nonzero: was it zero, then $|ak'-m\pi|$ would hold in view of \eqref{mi2}, contradicting thus the condition $(ak-m\pi)\cdot(bk-n\pi)<0$. Since $m,n,p,q\in\N$ by assumption, we have the trivial estimate $|np-mq|\geq1$. To sum up,
 % -------------- %
$$
\cotg ak\cotg bk<0\,\wedge\,|\sin bk|\geq2|\sin ak| \quad\Rightarrow\quad G(k)\geq\frac{9\pi}{2}\cdot\frac{1}{6p+\pi q}\,,
$$
 % -------------- %
which proves the first claim with $c=\frac{9\pi}{2(6p+\pi q)}$.

\smallskip

\noindent (ii) Let $\theta$ be badly approximable. Then $\theta':=1/\theta$ is badly approximable as well, i.e. there exists a $\gamma>0$ such that $\left|\theta'-\frac{n}{m}\right|>\frac{\gamma}{m^2}$ holds for all $n,m\in\N$. Using the estimate~\eqref{mi2} again, we obtain
 % -------------- %
\begin{align*}
G(k)&\geq\frac{9\pi}{2}\cdot\frac{|an-bm|}{6a+\pi b}=\frac{9\pi}{2} am\frac{\left|\frac{n}{m}-\frac{b}{a}\right|}{6a+\pi b}=\frac{9\pi}{2}\frac{am}{6a+\pi b}\left|\theta'-\frac{n}{m}\right| \\
&>\frac{9\pi}{2}\frac{am}{6a+\pi b}\cdot\frac{\gamma}{m^2}=\frac{9\pi}{2}\frac{\gamma}{6a+\pi b}\cdot\frac{a}{m}\,.
\end{align*}
 % -------------- %
We already know that $|ak-m\pi|\leq\frac{\pi}{6}$, hence $k\geq\frac{\pi}{a}\left(m-\frac{1}{6}\right)\geq\frac{\pi}{a}\cdot\frac{5m}{6}$. It means that $\frac{a}{m}\geq\frac{5\pi}{6}\cdot\frac{1}{k}$, which allows us to estimate $kG(k)$ as follows,
 % -------------- %
$$
kG(k)>\frac{9\pi}{2}\frac{\gamma}{6a+\pi b}\cdot\frac{5\pi}{6}\,.
$$
 % -------------- %
This yields the claim (ii) with $c=\frac{15\pi^2\gamma}{4(6a+\pi b)}$ concluding thus the proof.
\end{proof}

 % -------------- %
\begin{corollary}\label{coro: finitely}
Let $\theta=\frac{a}{b}$.
 % -------------- %
\begin{itemize}
 % -------------- %
\item[(i)] If $\theta\in\Q$, then the condition \eqref{GC2 b=c} generates at most finitely many gaps.
 % -------------- %
\item[(ii)] If $\theta$ is a badly approximable irrational, there exists a positive $\alpha_0$ such that the condition \eqref{GC2 b=c} generates no gaps for $0\le |\alpha|\le\alpha_0$.
\end{itemize}
\end{corollary}
 % -------------- %
\begin{proof}
According to Corollary \ref{coro b=c}, if $k$ is a solution of \eqref{GC2 b=c}, then $\frac{1}{|\sin ak|}\geq\frac{2}{|\sin bk|}$, $\,\cotg ak\cotg bk<0$, and $\left|G(k)-\frac{|\alpha|}{k}\right|\leq\frac{1}{|\sin ak|}-\frac{2}{|\sin bk|}$.

\smallskip

\noindent (i) Let $\theta\in\Q$. With regard to Lemma~\ref{Lemma yyy}, there exists a $c>0$ such that
 % -------------- %
$$
\frac{1}{|\sin ak|}>\frac{2}{|\sin bk|}\;\wedge\;\,\cotg ak\cotg bk<0 \quad\Rightarrow\quad G(k)\geq c
$$
 % -------------- %
holds for all $k>0$. Consequently, for $k\to\infty$ we have $G(k)>\frac{|\alpha|}{k}$. This allows us to remove the absolute value at the left-hand side of the condition $\left|G(k)-\frac{|\alpha|}{k}\right|\leq\frac{1}{|\sin ak|}-\frac{2}{|\sin bk|}$, which yields
 % -------------- %
\begin{equation}\label{rozepsana}
2\left(\frac{1}{|\sin bk|}-|\cotg bk|\right)-\left(\frac{1}{|\sin ak|}-|\cotg ak|\right)\leq\frac{|\alpha|}{k}\,.
\end{equation}
 % -------------- %
One can see, similarly as in the proof of Lemma~\ref{Lemma yyy}, that the left-hand side of \eqref{rozepsana} attains its local minima with respect to the condition $\frac{1}{|\sin ak|}\geq\frac{2}{|\sin bk|}\,\wedge\,\,\cotg ak\cotg bk<0$ at values $k$ satisfying $\frac{1}{|\sin ak|}=\frac{2}{|\sin bk|}$. This gives a necessary condition: Inequality \eqref{rozepsana} can be satisfied only if
 % -------------- %
$$
-2|\cotg bk|+|\cotg ak|\leq\frac{|\alpha|}{k}\,,
$$
 % -------------- %
i.e., for $G(k)\leq\frac{|\alpha|}{k}$. This is, however, impossible for $k\to\infty$, because $\frac{|\alpha|}{k}\to0$ and $G(k)\geq c>0$ due to the result of Lemma \ref{Lemma yyy}.

\smallskip

\noindent (ii) Let $\theta$ be badly approximable. In Lemma~\ref{Lemma yyy} we proved the existence of a $c>0$ such that for all $k>0$,
 % -------------- %
$$
\frac{1}{|\sin ak|}>\frac{2}{|\sin bk|}\;\wedge\;\,\cotg ak\cotg bk<0 \quad\Rightarrow\quad G(k)>\frac{c}{k}\,.
$$
 % -------------- %
In the rest of the proof we will demonstrate that one can set $\alpha_0:=c$.

Let us consider an $\alpha$ obeying $|\alpha|\leq\alpha_0:=c$. For such $\alpha$ we have $G(k)>\frac{c}{k}\geq\frac{|\alpha|}{k}$. Therefore, we can again remove the absolute value at the left-hand side of $\left|G(k)-\frac{|\alpha|}{k}\right|\leq\frac{1}{|\sin ak|}-\frac{2}{|\sin bk|}$, and obtain the condition \eqref{rozepsana}. Since the left-hand side attains its minimum with respect to the condition $\frac{1}{|\sin ak|}\geq\frac{2}{|\sin bk|}\,\wedge\,\,\cotg ak\cotg bk<0$ at $k$ satisfying $\frac{1}{|\sin ak|}=\frac{2}{|\sin bk|}$, it must hold $G(k)\leq\frac{|\alpha|}{k}$. However, for $|\alpha|<\alpha_0$ we have $G(k)>\frac{|\alpha|}{k}$ (see above), i.e., the last inequality cannot be fulfilled.
\end{proof}

 % -------------- %
\begin{theorem} \label{th: GC2}
Let $\theta=\frac{a}{b}$.
 % -------------- %
\begin{itemize}
 % -------------- %
\item[(i)] If $\theta$ is a Last admissible irrational number, then the condition \eqref{GC2 b=c} generates infinitely many gaps for any $\alpha\neq0$.
 % -------------- %
\item[(ii)] If $\theta$ is a badly approximable irrational, the condition \eqref{GC2 b=c} generates infinitely many gaps provided $|\alpha|\geq\frac{4\pi}{\sqrt{5}a}$.
\end{itemize}
 % -------------- %
\end{theorem}
 % -------------- %
\begin{proof}
We have shown that condition \eqref{GC2 b=c} is equivalent to $\frac{1}{|\sin ak|}>\frac{2}{|\sin bk|}$, $\,\alpha\cotg bk<0$, $\,\alpha\cotg ak<0$, and $G(k)<\frac{|\alpha|}{k}$ for $G$ given by equation~\eqref{G(k)}, see Corollary~\ref{coro b=c}. In particular, in the proof of Lemma~\ref{Lemma GC2} we have demonstrated that the system of conditions can be satisfied only for $ak\in\left(m\pi-\frac{\pi}{6},m\pi+\frac{\pi}{6}\right)$, i.e. for $k$ in certain neighbourhoods of $\frac{m\pi}{a}$.

\smallskip

\noindent (i) If $\theta$ is a Last admissible number, the same is true for $\theta':=1/\theta$. We can proceed in a way similar to the proof of Theorem~\ref{Proposition GC1}. There are integer sequences $\left\{p_n\right\}_{n=1}^\infty$ and $\left\{q_n\right\}_{n=1}^\infty$ such that $\lim_{n\to\infty}q_n^2\left|\theta'-\frac{p_n}{q_n}\right|=0$ and $\sgn\left(\theta'-\frac{p_n}{q_n}\right)=\sgn(\alpha)$. Since $\lim_{k\to\frac{q_n\pi}{a}}G(k)=\infty$, it holds $G(k)-\frac{|\alpha|}{k}>0$ in sufficiently small neighbourhoods of $\frac{q_n\pi}{a}$. Therefore, in small neighbourhoods of $\frac{q_n\pi}{a}$ the condition \eqref{GC2 b=c} acquires the form
 % -------------- %
\begin{equation}\label{rozepsana 2}
2\left(\frac{1}{|\sin bk|}-|\cotg bk|\right)-\left(\frac{1}{|\sin ak|}-|\cotg ak|\right)\leq\frac{|\alpha|}{k}\,;
\end{equation}
 % -------------- %
let us denote $2\left(\frac{1}{|\sin bk|}-|\cotg bk|\right)-\left(\frac{1}{|\sin ak|}-|\cotg ak|\right)=:W(k)$ for the sake of brevity.
Then
 % -------------- %
$$
q_n\lim_{k\to\frac{q_n\pi}{a}}W(k)
=2q_n\left(\frac{1}{\left|\sin\frac{b}{a}q_n\pi\right|}-\left|\cotg\frac{b}{a}q_n\pi\right|\right)
=2q_n\left|\tg\left(\left\{\frac{b}{a}q_n\right\}\frac{\pi}{2}\right)\right|\,.
$$
 % -------------- %
Since $\left|\left\{\frac{b}{a}q_n\right\}\frac{\pi}{2}\right| \leq\frac{\pi}{4}$ according to the definition~\eqref{symbol} and $|x|\leq|\tg x|\leq\frac{4}{\pi}|x|$ holds for $|x|\leq\frac{\pi}{4}$, we get
 % -------------- %
$$
q_n\lim_{k\to\frac{q_n\pi}{a}}W(k)
<4q_n\left|\left\{\frac{b}{a}q_n\right\}\right|
=4q_n^2\left|\theta'-\frac{\|\theta' q_n\|}{q_n}\right|
=4q_n^2\left|\theta'-\frac{p_n}{q_n}\right|\to 0
$$
 % -------------- %
as $n\to\infty$. At the same time we have
 % -------------- %
$$
q_n\lim_{k\to\frac{q_n\pi}{a}}\frac{|\alpha|}{k}=\frac{|\alpha|a}{\pi}=const.>0\,,
$$
 % -------------- %
and therefore inequality \eqref{rozepsana 2} is satisfied on a certain neigbourhood of $\frac{q_n\pi}{a}$ for every $n\in\N$. Let us check the remaining conditions from Corollary~\ref{coro b=c}. We will show that $\alpha\cotg bk>0$ and $\alpha\cotg ak<0$. The equation $\alpha\cotg bk>0$ is satisfied due to the choice of the sequence $\left\{q_n \right\}_{n=1}^\infty$. The equation $\alpha\cotg ak<0$ can be satisfied by choosing a \emph{left} (if $\alpha>0$) or \emph{right} (if $\alpha<0$) neigbhourhood of $\frac{q_n\pi}{a}$. To sum up, there are infinitely many integers $q_n\in\N$ such that the gap condition~\eqref{GC2 b=c} is satisfied for $k$ belonging to certain right or left neighbourhood of $\frac{q_n\pi}{a}$ for every $n\in\N$.

\smallskip

\noindent (ii) Let $\theta$ be an irrational number and $\theta'=1/\theta$. We shall demonstrate that if $|\alpha|\geq \frac{\pi^2}{\sqrt{5}a}$, then there are infinitely many $q\in\N$ such that $k$ in certain neighbourhoods of $\frac{q\pi}{a}$ satisfy the inequalities $\frac{1}{|\sin ak|}>\frac{2}{|\sin bk|}$ and $\left|G(k)-\frac{|\alpha|}{k}\right|<\frac{1}{|\sin ak|}-\frac{2}{|\sin bk|}$ together with the conditions $\alpha\cotg bk>0$ and $\alpha\cotg ak<0$. The first inequality is obviously valid for all $k$ sufficiently close to $\frac{q\pi}{a}$ with any $q\in\N$.
As for the second one, note that for $k$ sufficiently close to $\frac{q\pi}{a}$ it holds $G(k)>\frac{|\alpha|}{k}$, therefore, we shall prove that
 % -------------- %
$$
\lim_{k\to\frac{q\pi}{a}}W(k)<\lim_{k\to\frac{q\pi}{a}}\frac{|\alpha|}{k}
$$
 % -------------- %
for $W(k)$ introduced in part (i) above.
We have
 % -------------- %
$$
\lim_{k\to\frac{q\pi}{a}}W(k)
=2\left(\frac{1}{\left|\sin\frac{b}{a}q\pi\right|}
-\left|\cotg\frac{b}{a}q\pi\right|\right)
=2\left|\tg\left(\left\{\frac{b}{a}q\right\}\frac{\pi}{2}\right)\right|\,.
$$
 % -------------- %
Similarly as in part (i), we estimate the right-hand side of the last equation from above by $4\left|\left\{\frac{b}{a}q\right\}\right|$. For any irrational $\theta$ there are infinitely many $p,q\in\N$ such that $\left|\theta -\frac{p}{q}\right|<\frac{1}{\sqrt{5}q^2}$; in particular, for infinitely many $q\in\N$ it holds
 % -------------- %
$$
\left|\left\{\frac{b}{a}q\right\}\right|=\left|\theta q-\|\theta q\|\right|=q\left|\theta-\frac{\|\theta q\|}{q}\right|<q\frac{1}{\sqrt{5}q^2}=\frac{1}{\sqrt{5}q}\,.
$$
 % -------------- %
Consequently, for such $q$ we have
 % -------------- %
$$
\lim_{k\to\frac{q\pi}{a}}W(k)<\frac{4}{\sqrt{5}q}\,.
$$
 % -------------- %
On the other hand, $\lim_{k\to\frac{q\pi}{a}}\frac{|\alpha|}{k}=\frac{|\alpha|a}{q\pi}$, and therefore
 % -------------- %
$$
\lim_{k\to\frac{q\pi}{a}}W(k)<\lim_{k\to\frac{q\pi}{a}}\frac{|\alpha|}{k}\,,
$$
 % -------------- %
holds provided $|\alpha|\geq\frac{4\pi}{\sqrt{5}a}$. In other words, there are infinitely many $q\in\N$ such that the inequality $G(k) <\frac{|\alpha|}{k}$ is valid in a certain neighbourhood of $\frac{q\pi}{a}$.

Let us proceed to the condition $\alpha\cotg bk>0$. There are infinitely many $q\in\N$ such that $\left\{\frac{b}{a}q\right\}>0$ and infinitely many $q\in\N$ such that $\left\{\frac{b}{a}q\right\}<0$. Since $\sgn\left(\cotg b\frac{q\pi}{a}\right)=\sgn\left\{\frac{b}{a}q\right\}$, we conclude that inequality \eqref{rozepsana 2} and $\alpha\cotg bk>0$ can be satisfied simultaneously in certain ``Dirichlet point'' neighbourhoods for infinitely many $q\in\N$. Finally, the last condition $\alpha\cotg ak<0$ is obviously fulfilled in a sufficiently small \emph{left} (if $\alpha>0$) or \emph{right} (if $\alpha<0$) neighbourhood of $\frac{q\pi}{a}$ for any $q\in\N$.
\end{proof}

\section{Summary and open questions}

We have analyzed the spectrum of the quantum graph Hamiltonian describing a stretched hexagonal lattice with a $\delta$-coupling in the vertices, with a particular attention to the case when the stretch is parallel to one of the edges. In contrast to the case of a rectangular lattice \cite{Ex96, EG96} we have two different conditions determining the spectral gaps. They have nevertheless common features with respect to the number-theoretic properties of the lattice geometry, in particular, the existence of a critical coupling strengths needed to open spectral gaps in case of badly approximable edge lengths ratios.

Our results leave various questions open. An obvious one concerns the general case where we know that there are infinitely many open gaps for commensurate edges and $\alpha\ne 0$; once the commensurability hypothesis is abandoned we expect number-theoretic effect similar to those we have seen in the particular situation discussed in Sec.~\ref{s: b=c}. In addition to that, however, one wonders whether there are other number theoretic spectral features for $|\alpha|$ beyond the critical values similar to those observed in \cite{EG96}. In connection with Corollary~\ref{coro: finitely}(i) it would be also interesting to know whether one can have situations with a finite number of open gaps analogous to a Bethe-Sommerfeld-type spectrum of usual periodic Schr\"odinger operators in dimension two.

% ---------------------------------------------------------------------

\subsection*{Acknowledgment}
The work was supported by the Czech Science Foundation under the project 14-06818S.

% ------------------------------------------------------------------------

\end{document}